\documentclass{article}
\usepackage{amsmath}

\usepackage{amsthm}
\usepackage{sectsty}
\usepackage{bbm}
\usepackage{enumitem}
\usepackage{mathdots}
\usepackage{indentfirst}
\usepackage{tikz}
\usepackage{siunitx}
\usepackage{amssymb}
\usepackage{setspace}
\usepackage{amsfonts}
\usepackage{float}
\usepackage{graphicx}
\usepackage{multicol}
\usepackage{subcaption}
\usepackage[export]{adjustbox}
\usepackage[margin=1in]{geometry}
\sectionfont{\centering}
\newtheorem{theorem}{Theorem}[section]

\newtheorem{proposition}[theorem]{Proposition}
\newtheorem{corollary}[theorem]{Corollary}

\newtheorem{lemma}[theorem]{Lemma}

\title{Random Matching Markets with Correlated Preferences}
\author{Bill Wang}
\date{March 2026}
\begin{document}

\maketitle

\begin{abstract}
In the Gale-Shapley model of two-sided matching, it is well known that for generic preferences, the outcomes for each side can vary dramatically in the male-optimal vs. female-optimal stable matchings. In this paper, we show that under a widely used characterization of similarity in rankings, even a weak correlation in preferences guarantees assortative matching with high probability as the market size tends to infinity. It follows that the men's average ranking of women and the women's average ranking of men are asymptotically equivalent in all stable matchings with high probability, as long as the market imbalance is not too extreme.
\end{abstract}

\section{Introduction and Literature}
The study of two-sided matching initiated by \cite{gale1962college} has undoubtedly been one of the most important developments in modern economic theory. Its wide-scale adoption in various forms to school choice, residency placement, and organ exchange marks a significant achievement in bringing theory to bear on real life. Despite immense progress over the decades, many basic aspects of the theory are still poorly understood. One of the main obstacles to understanding its structure is the rapid explosion in the number of distinct preferences as the market size grows. For a given set of preferences, the collection of stable matchings is known to form a lattice, with a unique best stable matching for men and women respectively (when there are no ties). However, with $n$ men and $n$ women on each side of the market, the total number of distinct preferences is $(n!)^{2n}$. Even for $n=5$, this is a daunting figure, and there is no general theory that organizes the entire collection of stable matchings in terms of welfare. Thus, to make things tractable, a common approach is to study random matching markets, whereby the preferences on both sides are generated probabilistically, and analyze the likely outcomes in such large markets.
\medskip
\par
Several influential papers adopt this approach, but generally yield conclusions that are not entirely satisfactory. Most notably, \cite{pittel1989average} showed that in balanced markets where preferences are uniformly random, the men's average ranking of women in the male-optimal stable matching grows as $\log{n}$ with high probability, whereas in the female-optimal stable matching it grows as $n/\log{n}$ with high probability. This implies a multiplicity of equilibria and therefore substantial incentives for strategic behavior, which is at odds with empirical observations (e.g. \cite{roth1999redesign}). Maintaining the assumption of uniformly random preferences, \cite{ashlagi2017unbalanced} showed that adding a single individual to one side of the market leads to a near collapse of the core in favor of the short side and at the expense of the long side. This result, however, remains unnerving, because unbalanced markets abound in the real world, yet we seldom observe the kind of drastic differences in welfare that the theory predicts. In fact, most evidence points to assortative matching as the prevailing outcome (e.g. \cite{hitsch2010matching}).
\medskip
\par
It turns out that the assumption of uniform random preferences lies at the heart of these results. Yet it is difficult to find scenarios in practice where preferences over the full set of alternatives appear even remotely close to random. People may genuinely disagree as to whether Harvard is better than Princeton, for example, but few would object to the claim that both are better than say, Montana State. In this paper, we examine the set of stable outcomes in random matching markets when preferences by both genders exhibit correlation. Specifically, we assume that individual preferences independently follow a probability distribution on the set of permutations over members of the opposite sex. Correlation in preferences is parametrized by, and monotonically decreasing in, a coefficient $\phi\in[0,1]$, where a coefficient of $\phi=1$ corresponds to uniform random preferences and a coefficient of $\phi=0$ corresponds to unanimous preferences.
\medskip
\par
Our main result shows that when preferences are correlated, however weakly, assortative matching arises in large markets: for any fixed $\phi<1$, the quantile gap between any pair of stably matched man and woman vanishes in a balanced market with high probability as the market size tends to infinity. In doing so, we provide probabilistic bounds on the rank gap as a function of the correlation coefficients of both sides. As corollaries, we show that as long as the market imbalance does not explode faster than a polynomial rate, then a) For both genders, their average rank of matched partners is asymptotically equivalent over all stable matchings with high probability; b) The men's average ranking of women is also asymptotically equivalent to the women's average ranking of men with high probability. In other words, a slight amount of correlation in preferences is enough to guarantee agents be matched with someone of similar rank regardless of which side is doing the proposing and whether one is on the long side or short side of the market. This illustrates how the uniform random preference setting is really a `knife-edge' case that should be treated with caution.
\medskip
\par
This paper serves two purposes: First, it fills a void in the theoretical literature regarding the role of correlated preferences in the outcome of stable matchings. A recent paper (\cite{hoffman2023stable}) shows that the number of stable matchings grows exponentially in the market size with high probability whenever preferences follow a Mallows distribution, but does not characterize the structure of such matchings or address welfare implications. Attempts in understanding the latter have largely relied on simulation (e.g. \cite{boudreau2010marriage}) instead. Second, we see this paper as providing a micro foundation for assortative matching in the nontransferable utility framework. This contrasts with Becker's treatment of the subject in the transferable utility framework, where he showed that assortative matching arises when the surplus function is supermodular (\cite{becker1973theory}). While the theory of assortative matching is supported by ``complementarity'' of traits in potential couples when utility is transferable, a simple yet rigorous explanation of the phenomenon from the nontransferable utility side has been lacking. Given that matching in most social contexts lies somewhere in between the transferable and nontransferable worlds, this paper provides an alternative perspective on the issue by showing that stability and correlation in preferences, two conditions which arise naturally, necessarily implies assortative matching.

\section{Setup and Notation}

This section lays out the primitives of the model and the notation used throughout the remainder of this paper.
\medskip
\par
Let $M=\{m_{1},\ldots,m_{n}\}$ and $W=\{w_{1},\ldots,w_{l}\}$ be two disjoint sets denoting men and women respectively. Each man has preferences over the women and each woman has preferences over the men. A matching is a map $\mu:M\cup W\rightarrow M\cup W$ such that for every $m\in M$, either $\mu(m)=m$ or $\mu(m)\in W$ and for every $w\in W$, either $\mu(w)=w$ or $\mu(w)\in M$. In addition, if $\mu(x)=y$ for any $x\in M\cup W$, then $\mu(y)=x$. This is the one-to-one requirement. Without loss of generality, we assume $l=n+k$, where $k\geq 0$ captures the possibility of unbalanced markets. Furthermore, we impose the condition that every individual prefers to be paired up with someone of the opposite sex rather than remain single. A pair $(m,w)$ is said to block a matching $\mu$ if $m$ prefers $w$ over $\mu(m)$ and $w$ prefers $m$ over $\mu(w)$, and a matching without any blocking pairs is deemed stable.
\medskip
\par
All agents' preferences are assumed to independently follow a Mallows distribution \cite{mallows1957non}, which is defined over permutations and parametrized by a coefficient $\phi\in[0,1]$. For a Mallows distribution on $n$ elements with coefficient $\phi$, the probability that a permutation arises is proportional to $\phi^{\text{inv}(\pi)}$, where $\text{inv}(\pi)$ counts the number of inverted pairs in $\pi$. For example, if $\pi=(2,3,1)$ is the permutation that cyclically rotates the tuple, then $\text{inv}(\pi)=2$ because there are two inverted pairs: $(2,1)$ and $(3,1)$. More precisely, the $(\phi,n)$-Mallows measure $F_{\phi,n}$ over the set of permutations on $n$ elements $S_{n}$ is given by
\begin{equation*}
    F_{\phi,n}(\pi)=\frac{\phi^{\text{inv}(\pi)}}{Z_{\phi,n}},
\end{equation*}
where
\begin{equation*}
    \text{inv}(\pi)=|\{(i,j):i<j \text{ and } \pi(i)>\pi(j)\}|,
\end{equation*}
and
\begin{equation*}
    Z_{\phi,n}=\prod_{i=1}^{n}\frac{1-\phi^{i}}{1-\phi^{n}}
\end{equation*}
is a normalization constant.
\medskip
\par
As an example, suppose there are three women $\{w_{1},w_{2},w_{3}\}$, and the men's preferences over them follows the Mallows distribution $F_{0.5,3}$. This means that each man independently has preferences
    \begin{align*}
        & w_{1}\succ w_{2}\succ w_{3} \quad \text{with probability} \quad \frac{0.5^{0}}{Z};\\
        & w_{1}\succ w_{3}\succ w_{2} \quad \text{and} \quad w_{2}\succ w_{1}\succ w_{3},\quad \text{each with probability} \quad \frac{0.5^{1}}{Z};\\
        & w_{2}\succ w_{3}\succ w_{1} \quad \text{and} \quad w_{3}\succ w_{1}\succ w_{2},\quad \text{each with probability} \quad \frac{0.5^{2}}{Z};\\
        & w_{3}\succ w_{2}\succ w_{1} \quad \text{with probability} \quad \frac{0.5^{3}}{Z};
    \end{align*}
where $Z=21/8$ is a normalization constant that ensures the probabilities sum to $1$.
\medskip
\par
We allow men and women to have different degrees of correlation over the opposite sex, but each individual on the same side is parametrized by the same coefficient, denoted $\phi_{m}$ for men and $\phi_{w}$ for women. The Mallows central order is the reference order against which inversions are counted, namely, the identity permutation. In the context of college admissions, this can be thought of as a universal ranking that people use and broadly agree on, such as the one published by U.S. News \& World Report. We denote by $r(i)$ person $i$'s rank in the Mallows central order.

\section{Stable Matching under Correlated Preferences}

\subsection{Preliminaries}

Before proceeding towards the analysis, we derive a key property of $F_{\phi,n}$: With high probability, no element will be shifted by more than $O(\log{n})$ positions from its Mallows central order as $n\rightarrow\infty$. In fact, this statement remains true for a polynomial number of independent draws from a Mallows distribution:
\begin{lemma}\label{lem1}
    For every $j\in\{1,\ldots,s\}$, let $\pi_{j}\sim F_{\phi,t}$ independently. If $s\leq O(t^{z})$ for some $z\in\mathbb{R}_{+}$ and $d=c\log{t}$, where $c>-(z+1)/\log{\phi}$, then
    \begin{equation*}
        \lim_{t\rightarrow\infty}\mathbb{P}\bigg(\bigcap_{j=1}^{s}\bigcap_{i=1}^{t}\big\{|\pi_{j}(i)-i|<d\big\}\bigg)=1.
    \end{equation*}
\end{lemma}

\begin{proof}
Note that by independence, this is equivalent to showing that 
\begin{equation*}
    \lim_{t\rightarrow\infty}\mathbb{P}\bigg(\bigcap_{i=1}^{t}\big\{|\pi(i)-i|<d\big\}\bigg)^{s}=1
\end{equation*}
for a single $\pi\sim F_{\phi,n}$.
\medskip
\par
By Theorem $1$ of \cite{bhatnagar2015lengths}, for any $0<\phi<1$, $1\leq i\leq t$, and $d\geq 1$, we have
\begin{equation*}
    \mathbb{P}(|\pi(i)-i|\geq d)\leq 2\phi^{d}.
\end{equation*}

A union bound gives
\begin{equation*}
    \mathbb{P}\bigg(\bigcup_{i=1}^{t}\big\{|\pi(i)-i|\geq d\big\}\bigg)\leq 2t\phi^{d}.
\end{equation*}

Taking complements and raising powers:
\begin{equation*}
    \mathbb{P}\bigg(\bigcap_{i=1}^{t}\big\{|\pi(i)-i|<d\big\}\bigg)^{s}\geq (1-2t\phi^{d})^{s}.
\end{equation*}

Our goal is to pick $d$ appropriately as a function of $t$ so that the right-hand side approaches $1$ as $t$ tends to infinity. At the same time however, the left-hand side needs to be such that the concentration bounds remain meaningful. It turns out that a logarithm serves as the correct rate:
\medskip
\par
Letting $d=c\log{t}$, where $c>-(z+1)/\log{\phi}$, it is easy to see that $t\phi^{c\log{t}}\rightarrow 0$ as $t\rightarrow\infty$. Thus, by a Taylor expansion of $\log(1-x)=-x$ around $x=0$, we have
\begin{align*}
    \log[(1-2t\phi^{c\log{t}})^{s}]&=s\log(1-2t^{1+c\log{\phi}})\\
    &\approx s(-2t^{1+c\log{\phi}})\\
    &\geq -O(t^{z+1+c\log{\phi}}).
\end{align*}

It follows that
\begin{equation*}
    \lim_{t\rightarrow\infty}\log[(1-2t\phi^{c\log{t}})^{s}]=0\implies\lim_{t\rightarrow\infty}(1-2t\phi^{c\log{t}})^{s}=1.
\end{equation*}
\end{proof}

The lemma roughly says that as long as the market imbalance is not overly extreme, then with high probability, everyone on each side of the market will be ranked quite similarly by everyone of the opposite sex. This technical result will be used in a crucial manner later on.

To prove the assortative matching result, we utilize the fact that when preferences within each gender are highly similar, the mutual rank gap between any matched couple cannot differ by too much. This is formalized by Corollary $1$ of \cite{holzman2014matching}:
\medskip
\par
Given a set of rankings $R$ over $X$, define the displacement of $x\in X$ as $\delta(x)=\max_{r\in R}r(x)-\min_{r\in R}r(x)$ and let $\Delta(R)=\max_{x\in X}\delta(x)$ be the maximal displacement of the rankings $R$. Denote by $R_{M}$ and $R_{W}$ the set of men's preferences over women and women's preferences over men respectively. Furthermore, let $r_{m}(w)$ be $m$'s ranking of $w$ and $r_{w}(m)$ be $w$'s ranking of $m$ for every $m\in M$ and $w\in W$.
\begin{lemma}\label{lem2}
For any given set of preferences, 
\begin{equation*}
    |r_{m}(w)-r_{w}(m)|\leq 2\max\{\Delta(R_{W}),\Delta(R_{M})\}
\end{equation*}
under every stable matching $\mu$ and $(m,w)\in\mu$.
\end{lemma}

We embed this deterministic result into our probabilistic framework to show that:
\begin{proposition}\label{prop}
    Given $\phi_{m},\phi_{w}<1$, let the men's preferences over women and the women's preferences over men be distributed independently as $F_{\phi_{m},n+k}$ and $F_{\phi_{w},n}$ respectively, where $k$ is at most polynomial in $n$. On a sequence of events $\{E_{n}\}_{n\in\mathbb{N}}$, with $\mathbb{P}(E_{n})\rightarrow 1$, we have $|r_{w}(m)-r_{m}(w)|\leq O(\log(n+k))$ and $|r(m)-r(w)|\leq O(\log(n+k))$ for every stable matching $\mu$ and $(m,w)\in\mu$.
\end{proposition}

The idea is that no matter how weak the correlation in preferences is, disparities in popularity add up over long distances, so as the market size tends to infinity, the probability that anyone is displaced very far from the Mallows central order is very small. Thus, markets effectively become localized, and the gap in mutual rankings does not grow too fast. The key technical contribution is in showing that this rate is bounded by $O(\log(n+k))$ with high probability.

\begin{proof}
Let $A_{n}$ denote the event that no man shifts any woman by $c_{m}\log(n+k)$ or more ranks away from their position in the Mallows central order, where $c_{m}>-(z+1)/\log{\phi_{m}}$ and $B_{n}$ denote the event that no woman shifts any man by $c_{w}\log{n}$ or more ranks away from their position in the Mallows central order, where $c_{w}>-2/\log{\phi_{w}}$. It follows that the maximal displacement of women by men is $\Delta(R_{M})=2c_{m}\log(n+k)$ on $A_{n}$ and the maximal displacement of men by women is $\Delta(R_{W})=2c_{w}\log{n}$ on $B_{n}$.  By Lemma \ref{lem1}, we know that $\mathbb{P}(A_{n}),\mathbb{P}(B_{n})\rightarrow 1$ as $n\rightarrow\infty$. Thus, on the event $E_{n}=A_{n}\cap B_{n}$, where $\mathbb{P}(E_{n})=\mathbb{P}(A_{n})\mathbb{P}(B_{n})\rightarrow 1$, we have
\begin{equation*}
    |r_{w}(m)-r_{m}(w)|\leq 2\max\{2c_{m}\log(n+k),2c_{w}\log{n}\}=O(\log(n+k))
\end{equation*}
for every stable matching $\mu$ and $(m,w)\in\mu$ by Lemma \ref{lem2}.
\medskip
\par
Furthermore, note that $|r(w)-r_{m}(w)|$ and $|r_{w}(m)-r(m)|$ are by definition both bounded by $O(\log(n+k))$ for all $m\in M$ and $w\in W$ on $E_{n}$. Hence, by the triangle inequality, we have
\begin{equation}\label{eq:1}
    |r(w)-r(m)|\leq|r(w)-r_{m}(w)|+|r_{m}(w)-r_{w}(m)|+|r_{w}(m)-r(m)|=O(\log(n+k))
\end{equation}
for every stable matching $\mu$ and $(m,w)\in\mu$.
\end{proof}

\subsection{Assortative Matching}

We are now in a position to prove the main result of the paper, which says that in a balanced market, each person is matched to someone of similar rank in every stable matching with high probability as the market size becomes large. Define person $i$'s quantile rank in the population as $q_{i}=r(i)/n$. This is their relative ranking against members of the same sex according to the Mallows central order.

\begin{theorem}
    Given $\phi_{m},\phi_{w}<1$, let the men's preferences over women and the women's preferences over men be distributed independently as $F_{\phi_{m},n}$ and $F_{\phi_{w},n}$ respectively. On a sequence of events $\{E_{n}\}_{n\in\mathbb{N}}$, with $\mathbb{P}(E_{n})\rightarrow 1$, the quantile gap $|q_{w}-q_{m}|$ of every stably matched pair $(m,w)$ converges to zero as the market size $n$ tends to infinity.
\end{theorem}

\begin{proof}
The sequence of events $\{E_{n}\}_{n\in\mathbb{N}}$ is defined as in Proposition \ref{prop}. Dividing (\ref{eq:1}) by $n$ and taking limits:
\begin{equation*}
    \lim_{n\rightarrow\infty}\bigg|\frac{r(w)}{n}-\frac{r(m)}{n}\bigg|\leq\lim_{n\rightarrow\infty}\frac{O(\log{n})}{n}=0.
\end{equation*}

This shows that the gap in rank quantiles vanishes in the limit.
\end{proof}

\subsection{Welfare Equivalence}

Next, we turn to welfare implications of stable matchings with correlated preferences, which hold even under market imbalance. For a given matching $\mu$, we define the men's average ranking of women as
\begin{equation*}
    A_{M}(\mu)=\frac{1}{|M\setminus\overline{M}|}\sum_{m\in M\setminus\overline{M}}r_{m}(\mu(m)),
\end{equation*}
where $\overline{M}$ is the set of men who are unmatched under $\mu$, and the women's average rank of men as
\begin{equation*}
    A_{W}(\mu)=\frac{1}{|W\setminus\overline{W}|}\sum_{w\in W\setminus\overline{W}}r_{w}(\mu(w)),
\end{equation*}
where $\overline{W}$ is the set of women who are unmatched under $\mu$. Note that under our assumptions, the men are always matched, whereas exactly $k$ women remain unmatched. The average ranking is taken as a measure of the welfare of each side.
\medskip
\par
Our first corollary shows the equivalence in welfare between the optimal vs. pessimal matching for each side as the market size $n$ tends to infinity. Denote by $\mu_{M}$ and $\mu_{W}$ the male and female optimal matching respectively.

\begin{corollary}\label{cor}
    Given $\phi_{m},\phi_{w}<1$, if the men's preferences over the women and the women's preferences over the men are distributed independently as $F_{\phi_{m},n+k}$ and $F_{\phi_{w},n}$ respectively, where $k$ is at most polynomial in $n$, then on a sequence of events $\{E_{n}\}_{n\in\mathbb{N}}$, with $\mathbb{P}(E_{n})\rightarrow 1$, we have
    \begin{equation*}
        \lim_{n\rightarrow\infty}\frac{A_{M}(\mu_{W})}{A_{M}(\mu_{M})}=\lim_{n\rightarrow\infty}\frac{A_{W}(\mu_{M})}{A_{W}(\mu_{W})}=1.
    \end{equation*}
\end{corollary}

\begin{proof}
The sequence of events $\{E_{n}\}_{n\in\mathbb{N}}$ is defined as in Proposition \ref{prop}, from which we have that
\begin{equation*}
    |r_{m}(w)-r(m)|\leq |r_{m}(w)-r_{w}(m)|+|r_{w}(m)-r(m)|=O(\log(n+k)),
\end{equation*}
for any stable matching $\mu$ and $(m,w)\in\mu$. Thus, $r_{m}(w)=r(m)\pm O(\log(n+k))$, so that
\begin{equation*}
    \frac{\sum_{m\in M}[r(m)-O(\log(n+k))]}{n}\leq A_{M}(\mu)\leq\frac{\sum_{m\in M}[r(m)+O(\log(n+k))]}{n}
\end{equation*}
on $E_{n}$. Observing that $\sum_{m\in M}r(m)=1+\ldots+n=n(n+1)/2$, we have
\begin{equation*}
    \lim_{n\rightarrow\infty}\frac{A_{M}(\mu_{W})}{A_{M}(\mu_{M})}\leq\lim_{n\rightarrow\infty}\frac{(n+1)/2+O(\log(n+k))}{(n+1)/2-O(\log(n+k))}=1.
\end{equation*}

Since $A_{M}(\mu_{W})\geq A_{M}(\mu_{M})$, it follows that $\lim_{n\rightarrow\infty}A_{M}(\mu_{W})/A_{M}(\mu_{M})=1$. The same reasoning applies to the women's average ranking.
\end{proof}

This means that which side does the proposing makes little difference to welfare when preferences are correlated. This is in sharp contrast to the uniform random preference case (in balanced markets), where by \cite{pittel1989average},
\begin{equation*}
    \frac{A_{M}(\mu_{W})}{A_{M}(\mu_{M})}=\frac{A_{W}(\mu_{M})}{A_{W}(\mu_{W})}\sim\frac{n}{(\log{n})^{2}}\rightarrow\infty
\end{equation*}
with high probability.
\medskip
\par
Our second corollary shows the equivalence of welfare for men and women in the presence of market imbalance:

\begin{corollary}
    Given $\phi_{m},\phi_{w}<1$, if the men's preferences over the women and the women's preferences over the men are distributed independently as $F_{\phi_{m},n+k}$ and $F_{\phi_{w},n}$ respectively, where $k$ is at most polynomial in $n$, then on a sequence of events $\{E_{n}\}_{n\in\mathbb{N}}$, with $\mathbb{P}(E_{n})\rightarrow 1$, we have
    \begin{equation*}
        \lim_{n\rightarrow\infty}\frac{A_{W}(\mu)}{A_{M}(\mu)}=1.
    \end{equation*}
    in every stable matching $\mu$.
\end{corollary}

\begin{proof}
    The sequence of events $\{E_{n}\}_{n\in\mathbb{N}}$ is defined as in Proposition \ref{prop}, from which we have
    \begin{equation*}
        A_{W}(\mu)=\frac{\sum_{m\in M}r_{\mu(m)}(m)}{n}\leq\frac{\sum_{m\in M}r_{m}(\mu(m))+O(\log(n+k))}{n}=A_{M}(\mu)+O(\log(n+k)),
    \end{equation*}
    for every stable matching $\mu$ on $E_{n}$.
    \medskip
    \par
    Now from the inequality $A_{M}(\mu)\geq (n+1)/2-O(\log(n+k))$ derived in Corollary \ref{cor}, it follows that
    \begin{align*}
        \frac{A_{W}(\mu)}{A_{M}(\mu)}&\leq\frac{A_{M}(\mu)+O(\log(n+k))}{A_{M}(\mu)}\\
        &\leq 1+\frac{O(\log(n+k))}{(n+1)/2-O(\log(n+k))}.
    \end{align*}
    
    Thus, $\lim_{n\rightarrow\infty}A_{W}(\mu)/A_{M}(\mu)\leq 1$, and by symmetry, $\lim_{n\rightarrow\infty}A_{M}(\mu)/A_{W}(\mu)\leq 1$. The result follows.
\end{proof}

This means that there is no advantage to the short side and no disadvantage to the long side when preferences are correlated. This is in sharp contrast to the uniform preference case (in unbalanced markets), where by \cite{ashlagi2017unbalanced},
\begin{equation*}
    \lim_{n\rightarrow\infty}\frac{A_{W}(\mu)}{A_{M}(\mu)}=\infty
\end{equation*}
with high probability in every stable matching $\mu$.

\newpage
\bibliographystyle{apalike}
\bibliography{reference}

\end{document}